\documentclass[12pt]{amsart}

\usepackage[matrix,arrow,curve,frame]{xy}
\usepackage{amsmath,amsthm,amssymb,enumerate}
\usepackage{latexsym}
\usepackage{amscd}
\usepackage{euscript}
\usepackage{fullpage}
\usepackage{color}

\usepackage[all]{xypic}

\newtheorem{theorem}{Theorem}[section]
\newtheorem{lemma}[theorem]{Lemma}
\newtheorem{corollary}[theorem]{Corollary}

\theoremstyle{definition}

\theoremstyle{remark}
\newtheorem{remark}[theorem]{Remark}

\numberwithin{equation}{section}
\newcommand{\beq}{\begin{equation}}
\newcommand{\eeq}{\end{equation}}

\newcommand{\TT}{\mathbb{T}}
\newcommand{\ZZ}{\mathbb{Z}}

\newcommand\bbT{\mathbb T}

\newcommand{\U}{\mathcal{U}}
\newcommand{\V}{\mathcal{V}}

\renewcommand{\cL}{\mathcal{L}}
\newcommand{\cG}{\mathcal{G}}

\newcommand{\intomega}{\int^{P/M, n}\!\!\!\!\!\omega}

\newcommand {\be}{\begin{equation}}
\newcommand {\ee}{\end{equation}}
\newcommand{\h}{\begin{eqnarray*}}
\newcommand{\e}{\end{eqnarray*}}

\begin{document}


\title
{T-duality in an H-flux: \\ Exchange of momentum and winding\\}


\author{Fei Han}
\address{Department of Mathematics,
National University of Singapore, Singapore 119076}
\email{mathanf@nus.edu.sg}

 \author{Varghese Mathai}
\address{School of Mathematical Sciences,
University of Adelaide, Adelaide 5005, Australia}
\email{mathai.varghese@adelaide.edu.au}

\date{}

\maketitle

\begin{abstract}
Using our earlier proposal for Ramond-Ramond fields in an H-flux on loop space \cite{HM15}, {we extend the Hori isomorphism in
\cite{BEM04a,BEM04b} from invariant differential forms, to invariant exotic differential forms} such that the {\em momentum} and {\em winding numbers} are exchanged,
filling in a gap in the literature.  We { also extend the compatibility of the action of invariant exact Courant algebroids on the
T-duality isomorphism in \cite{CG}, to the
T-duality isomorphism on exotic invariant differential forms. }

\end{abstract}

\tableofcontents

\section*{Introduction}
In \cite{BEM04a,BEM04b}, T-duality in a background flux was studied for the first time, and we summarise it here to begin with. {{Upon compactifying 
spacetime in one direction to a principal circle bundle $\TT \to Z \stackrel{\pi}{\to} X$ with background { $\TT$-invariant }flux $H$ which is a closed 3-form on $Z$.}} Then 
there is a T-dual circle bundle $\hat\TT\to\hat Z \stackrel{\hat\pi}{\to} X$ with T-dual background { $\hat \TT$-invariant flux $\hat H$} which is a closed 3-form on $\hat Z$ such that $c_1(Z) = \hat\pi_*[\hat H]$ and $c_1(\hat Z) = \pi_*[ H]$, and the constraint that $[H]=[\hat H]$ on the correspondence space
$Z\times_X \hat Z$  ensures that $[\hat H]$ is uniquely defined. The slogan that, 
\begin{center}
{\em the Chern class is exchanged with the background flux} 
\end{center}
encapsulates T-duality in a background flux, so there is a change in topology if either the Chern class or the background flux is topologically nontrivial.

Choosing { $\TT$-invariant connection 1-forms $A$ on $Z$ and $\hat \TT$-invariant connection $\hat A$ on $\hat Z$ respectively}, the rules for transforming the RR fields can be encoded in the 
twisted Fourier-Mukai transform \cite{BEM04a, BEM04b}
 \begin{equation} \label{eqAh}
T_*G =  \int_{\bbT}  G \, e^{ -A \wedge \hat A }, 
\end{equation}
where $G \in \Omega^\bullet(Z)^\bbT$ is an invariant differential form on spacetime which is the total Ramond-Ramond field-strength,
\begin{center}
$G\in\Omega^{\bar k}(Z)^\bbT \quad$ for { Type IIA};\\
$G\in\Omega^{\overline{k+1}}(Z)^\bbT \quad$ for { Type IIB},\\
\end{center}
for $\bar k= k \mod 2$, and where the right hand side of equation \eqref{eqAh} is an invariant differential form on $Z\times_X\hat Z$, and 
the integration is along the $\bbT$-fiber of $Z$. 

Define the Riemannian metrics on $Z$ and $\hat Z$ by
$$  
g_Z=\pi^*g_X+R^2 A\odot A,\qquad g_{\hat Z}=\hat\pi^*g_X+ R^{-2}\hat A\odot\hat A.
$$
Under the above choices of Riemannian metrics and flux forms, the twisted Fourier-Mukai transform is an isometry
\begin{equation} \label{eqT-duality}
T\colon\Omega^{\bar k}(Z)^\TT\to\Omega^{\overline{k+1}}(\hat Z)^{\hat\TT},
\end{equation}
for $\bar k= k \mod 2$,  inducing isometries on the spaces of twisted
harmonic forms. 

In particular, $R$ goes to $1/R$, which is an important feature of T-duality, and there is an induced degree-shifting isomorphism
between twisted cohomology groups,
\be
T_* : H^{\bar k}(Z, H) \cong H^{\overline{k+1}}(\hat Z, \hat H).
\ee
for $\bar k= k \mod 2$, { where $H^*(Z, H):=\{\Omega^*(Z)^{\TT}, d+H)$ and $H^*(\hat Z, H):=\{\Omega^*(\hat Z)^{\hat \TT}, d+\hat H).$} These twisted cohomology groups were first defined in \cite{RW} and their relation to 
D-branes in an H-flux and their charges were further explored in \cite{BCMMS,MS}.

One of the deficiencies of the T-duality isomorphism in equation \eqref{eqT-duality} is that it is only 
defined on the smaller configuration space of invariant differential forms on spacetime $Z$, and therefore is not easy to formulate the exchange of momentum 
and winding in this framework. We rectify this in our paper as follows.

Let $\xi, \hat\xi$ denote the complex line bundles associated to the circle bundles $Z, \hat Z$ and the standard representation of the circle on complex plane respectively. Define the {\bf exotic differential forms} by
$$\mathcal{A}^{\bar k}(Z)^{\TT}=\bigoplus_{n\in \ZZ}\mathcal{A}^{\bar k}_n(Z)^{\TT}:=\bigoplus_{n\in \ZZ}\Omega^{\bar k}(Z, \pi^*(\hat{\xi}^{\otimes n}))^{\TT},$$
$$\mathcal{A}^{\bar k}(\hat{Z})^{\hat\TT}=\bigoplus_{n\in \ZZ}\mathcal{A}^{\bar k}_n(\hat{Z})^{\hat\TT}:=\bigoplus_{n\in \ZZ}\Omega^{\bar k}(\hat{Z}, \hat{\pi}^*(\xi^{\otimes n}))^{\hat\TT}$$
for $\bar k= k \mod 2$.for $\bar k= k \mod 2$. 

An analogous space of exotic differential forms and equivariantly flat superconnection \cite{MQ} was
  first defined on loop space in our earlier paper \cite{HM15}, which was inspired by and generalises some of the results 
  in \cite{B85}, and we now briefly outline here. It motivates
  the definition of the spaces of exotic differential forms above.

  Let $(H, B_\alpha, F_{\alpha\beta}, L_{\alpha\beta})$ denote a gerbe with connection on $Z$ (cf. \cite{Brylinski}), where 
$(H, B_\alpha, F_{\alpha\beta})$ denotes the Deligne class of the closed integral 3-form $H$ with respect to 
a Brylinski open cover (cf. \cite{HM15}) and $ L_{\alpha\beta}$
denotes the line bundles on double overlaps that determines the gerbe $\cG$. The holonomy of the gerbe is then a line 
bundle $\cL$ with connection $\tau(B_\alpha)$ having curvature $\tau(H)$ on loop space $LZ$, where $\tau$ denotes the transgression map.
It $\iota_0 : Z \to LZ$ denotes the embedding of $Z$ into $LZ$ as the constant loops, then noting that $\iota_0^*(\cL)$ is
canonically trivial, we proved that 
{ 
\be \label{degree0}  \iota_0^*: \Omega^{\bar k}(LZ, \cL)^{S^1} \longrightarrow \Omega^{\bar k}(Z) \ee}intertwines the equivariantly flat superconnection
 on the left hand side and $d+H$ on the right hand side, where the left hand side was called there the exotic differential forms on 
 loop space. 
 
 The precise relation between \cite{HM15} and this paper is that when $Z$ is the total space of a principal circle bundle, then
  there is a natural infinite sequence of embeddings $\iota_n : Z \to LZ$ defined by $\iota_n(x): S^1\ni t\mapsto  \gamma_x(t)=t^n\cdot x$, for all $n\in \ZZ$. We consider such sequence of embeddings motivated by the fact that there are $\mathbb{Z}$ many connected components 
  in the loop space $L\TT$. 
  We have 
  $\iota_n^*(\cL) \cong  \pi^*(\hat\xi)^{\otimes n}$ since they have the same Chern class. The loop space $LZ$ has the natural circle action by rotating loops, and $Z$ has a circle action as the total space of circle bundle. To tell the difference of these two circle actions, we use $S^1$ for the  circle action by rotating loops. 
{ We have that for $n\neq 0$, 
$$\iota_n^*: \Omega^{\bar k}(LZ, \cL)^{S^1} \longrightarrow \Omega^{\bar k}(Z, \pi^*(\hat{\xi}^{\otimes n}))^{\TT}$$ intertwines the equivariantly flat superconnections on both spaces. We would like to point out that we don't automatically see  the $\TT$-invariance things on $Z$ for $n=0$ in the map (\ref{degree0}). This comes from the effects of our embedding $\iota_n$. However this point of view does motivate us to develop the exotic theories on $Z$ and eventually relates to the Fourier expansion as discussed below.}

{ Define the subspace of weight $-n$ differential forms on $Z$,
\be \Omega^*_{-n}(Z)=\{\omega \in \Omega^*(Z)| L_{v}\omega=-n\omega\}.\ee 
It is easy to see that
$$\Omega^{\bar k}_0(Z)=\Omega^{\bar k}(Z)^\TT, \ \  \mathcal{A}^{\bar k+1}_0(\hat{Z})^{\hat\TT}=\Omega^{\overline{k+1}}(\hat Z)^{\hat\TT}.$$
Under the above choices of Riemannian metrics and flux forms, our results show that the Fourier-Mukai transform $T$ can be extended to a sequence of  isometries,
\begin{equation}
\tau_n \colon \Omega^{\bar k}_{-n}(Z) \to \mathcal{A}^{\bar k+1}_n(\hat{Z})^{\hat\TT},
\end{equation}
for $\bar k= k \mod 2$, and is defined by the {\em exotic Hori formula} from $Z$ to $\hat Z$ given in equation \eqref{exotic-Hori}. When $n=0$, $\tau_0=T$. Theorem \ref{main} shows that the twisted de Rham differential $d+H$ maps to the differential $-(\hat{\pi}^*\nabla^{\xi^{\otimes n}}-\iota_{n\hat{v}}+\hat{H})$.  { One similarly has 
a sequence of  isometries,
\begin{equation}
\sigma_n \colon\mathcal{A}^{\bar k}_n(Z)^{\TT} \to \Omega^{\bar k +1}_{-n}(\hat{Z}),
\end{equation}
for $\bar k= k \mod 2$, and is defined by the {\em inverse exotic Hori formula} form $Z$ to $\hat Z$ given in equation \eqref{exotic-Hori2} and the differential $\pi^*\nabla^{\hat\xi^{\otimes n}}-\iota_{n{v}}+{H}$
maps to  the twisted de Rham differential $-(d+{\hat H})$. Note that $\sigma_0=T$. Similarly, one can define the sequences of isometries $\hat\tau_n, \hat\sigma_n$ on $\hat Z$. Although the extension of the Fourier-Mukai transform 
to all differential forms on $Z$ is slightly asymmetric, one has the crucial identities, verified in Theorem \ref{main},
\begin{align}
{\rm -Id} =\hat\sigma_n \circ \tau_n \colon \Omega^{\bar k}_{-n}(Z)  \longrightarrow \Omega^{\bar k}_{-n}(Z),\\
{\rm -Id} =\hat\tau_n \circ \sigma_n  \colon \mathcal{A}^{\bar k}_n(Z)^{\TT} \longrightarrow \mathcal{A}^{\bar k}_n(Z)^{\TT}.  
\end{align}
This is interpreted as saying that T-duality when applied twice, returns one to minus of the identity. It was verified in the special case when $n=0$ in \cite{BEM04a,BEM04b}. We would like to point out that the minus sign comes from the convention of integration along the fiber. 

This shows that for each of either $Z$ or $\hat Z$, there are two theories (at degree 0 the two theories coincide), and there are also graded isomorphisms between the two theories of both sides. }

Theorem \ref{concentration} tells us that when $n\neq 0$,  
the complex $(\mathcal{A}^{\bar k+1}_n(\hat{Z})^{\hat\TT}, \hat{\pi}^*\nabla^{\xi^{\otimes n}}-\iota_{n\hat{v}}+\hat{H})$ has vanishing cohomology. Therefore, when $n\neq 0$, the complex $(\Omega^{\bar k}_{-n}(Z), d+H)$ also has 
vanishing cohomology. In Corollary \ref{nullcoho}, we construct a homotopy to show this by taking advantage of the homotopy operator previously constructed in Theorem \ref{concentration}.

For a general form $\omega\in \Omega^*(Z)$, not necessarily $\TT$-invariant, one can perform the family Fourier expansion (see Section 2.4) to get 
$$
 \Omega^*(Z) \ni \omega=\sum_{n=-\infty}^\infty\omega_n \in \bigoplus_{n=-\infty}^\infty \Omega^*_{n}(Z)
 $$ 
with $\omega_n\in \Omega^{\bar k}_{n}(Z)$ and one takes the Fr\'echet space completion of the direct sum above.  Since $H$ is $\TT$-invariant, if $(d+H)\omega=0$, then $(d+H)\omega_n=0$.  { Corollary \ref{nullcoho}} shows that $\omega_n$ must be $(d+H)$-exact when $n\neq 0$. Therefore the cohomology of $(\Omega^*(Z), d+H)$ is concentrated on the $n=0$ part, i.e. the $\TT$-invariant part. This indicates why in the previous study, we only consider the $\TT$-invariant forms. 
}

Now we are able to define momentum and winding in the much larger configuration space of invariant exotic differential forms, $\mathcal{A}^{\bar k}(Z)^{\TT}$
 as follows. The multiple $w$ of the infinitesimal generator $v$ of the circle action on $Z$ is the {\em winding}, as it agrees with 
winding around the circle direction when the circle bundle is a product, cf. \cite{Hori}. The tensor power $m$ of the line bundle ${\xi}$  is the {\em momentum}, as it agrees with 
momentum when the circle bundle is a product, cf. \cite{Hori}. In Theorem \ref{flat} and section 2.4, we show that
the momentum on the spacetime $Z$ needs to be equal to the winding number on the T-dual spacetime $\hat{Z}$, 
in order that  the exotic differential is an equivariantly flat superconnection.  The T-dual side also exhibits this phenomena. 
Thus our slogan here is, \\

 \begin{quote}{\em The momentum (on spacetime $Z$) gets exchanged with the winding number (on the T-dual spacetime $\hat{Z}$)
and the winding number (on the spacetime $Z$) gets exchanged with the momentum (on the T-dual spacetime $\hat{Z}$)}.\\
\end{quote}

Finally, the invariant exact Courant algebroid $(TZ\oplus T^*Z)^\TT_H$, together with the usual Dorfman bracket, $u\circ_H v$, has a representation (or action) on the exotic differential forms via a novel exotic Lie derivative  $L_{X+\alpha}^\xi$ (c.f. Theorem \ref{thm:corant-relations}) and where the Courant bracket is still as before, namely the antisymmetrization of the Dorfmann bracket.
For the relationship between invariant Courant algebroids and T-duality, see \cite{CG}.
\bigskip

\noindent{\em Acknowledgements.}
The first author was partially supported
by the grant AcRF R-146-000-218-112 from National University
of Singapore. The second author was partially supported by funding from the Australian Research Council, through Discovery Project grant DP150100008 and by the Australian Laureate Fellowship  FL170100020. 
The first author also thanks Qin Li for helpful discussions.
The second author thanks Maxim Zabzine (Uppsala), who posed the question (private communication) that is answered in this paper.
Both authors thank Peter Bouwknegt for sharing some insights, and also the
 Chern Institute of Mathematics (Tianjin) for providing a stimulating research environment on our visit.

\section{Exact Courant algebroids and exotic differential forms}


In this section, we consider invariant exact Courant algebroids and their actions on invariant differential forms with coefficients in a line bundle.

Let $M$ be a smooth manifold. Consider the generalized tangent bundle $\mathcal{T}M=TM\oplus T^*M$. On sections of $\mathcal{T}M$, there is a natural field of non-degenerate symmetric bilinear form, namely for $X+\alpha, Y+\beta\in \Gamma(TM\oplus T^*M)$, we put
\be \langle X+\alpha, Y+\beta\rangle=\frac{1}{2}(\iota_X\beta+\iota_Y\alpha). \ee

The Clifford algebra $\mathrm{Cliff}(TM\oplus T^*M)$ is the algebra with generators $\gamma_\U,\, \U\in \Gamma(TM\oplus T^*M)$ and relations
\be \{ \gamma_\U, \gamma_\V\}=2\langle \U, \V\rangle. \ee

Further assume that $M$ admits a smooth $\TT$-action and let $\xi$ be a $\TT$-equivariant complex line bundle over $M$. Let $\nabla^\xi$  be a $\TT$-invariant connection on $E$  and $H\in \Omega^3_{cl}(M)$ a closed 3-form such that 
\be \label{maineqn} (\nabla^{\xi}-\iota_v+H)^2+L_v^\xi=0,\ee 
where $v$ is the Killing vector field of the $\TT$-action. $\nabla^{\xi}-\iota_v+H$ and $L_v^\xi$ are operators acting on $\Omega^*(M,\xi)$, the space of smooth differential forms with coefficients in $\xi$. 

It is known that the Clifford algebra $\mathrm{Cliff}(TM\oplus T^*M)$ has a natural action on $\Omega^*(M)$. One can easily extend this action to $\Omega^*(M,\xi)$.
\begin{lemma} We have a representation of the Clifford algebra $\mathrm{Cliff}(TM\oplus T^*M)$ on $\Omega^*(M,\xi)$ given by
$$\gamma_{X+\alpha}\cdot \varphi=\iota_X\varphi+\alpha\wedge\varphi, \ \ \ \ X+\alpha\in \Gamma(TM\oplus T^*M), \ \ \ \varphi \in \Omega^*(M,\xi).$$
\end{lemma} 

For $X+\alpha, Y+\beta\in \Gamma(TM\oplus T^*M)$ and $H\in \Omega^3_{cl}(M)$, define the (twisted) Dorfmann bracket or Loday bracket by
\be (X+\alpha)\circ_H (Y+\beta)=[X,Y]+L_X\beta-\iota_Yd\alpha+\iota_X\iota_YH.
\ee
It is related to the (twisted) Courant bracket by 
\be \U\circ_H \V=[[\U, \V]]_H+d\langle \U, \V\rangle, \ \ \U, \V\in \Gamma(TM\oplus T^*M),  \ee or conversely,
\be [[\U, \V]]_H=\frac{1}{2}(\U\circ_H \V-\V\circ_H \U).\ee

For $\U=X+\alpha\in \Gamma(TM\oplus T^*M)$, we introduce an {\bf exotic twisted Lie derivative along $\U$} on $\Omega^*(M,\xi)$ by
\be \mathcal{L}^\xi_{\U}=L_X^\xi-\mu^\xi_X+(d\alpha-\iota_v\alpha+\iota_X H)\wedge,\ee
where $L_X^\xi$ is the Lie derivative along the direction $X$ and $\mu^\xi_X$ is the moment of the $\TT$-invariant connection $\nabla^\xi$ along the direction $X$.  Evidently, $\mathcal{L}^\xi_{\U}$ depends on $\nabla^\xi, v$ and $H$.

We have the following relations.  
\begin{theorem}\label{thm:corant-relations} Let \,$\U, \V\in \Gamma(TM\oplus T^*M)$. Then on  $\Omega^*(M,\xi)$, we have 
\be
\begin{split}
&\{\gamma_\U, \gamma_\V\}=2\langle \U, \V\rangle,\\
&\{\nabla^{\xi}-\iota_v+H, \gamma_\U\}=\mathcal{L}^\xi_{\U},\\
& [\nabla^{\xi}-\iota_v+H, \mathcal{L}^\xi_{\U}]=0\ \mathrm{on}\ \Omega^*(M,\xi)^{\TT},\\
&[\mathcal{L}^\xi_{\U}, \gamma_\V]=\gamma_{\U\circ_H \V},\\
& [\mathcal{L}^\xi_{\U}, \mathcal{L}^\xi_{\V}]=\mathcal{L}^\xi_{\U\circ_H \V}=\mathcal{L}^\xi_{[[\U, \V]]_H} \mathrm{on}\ \Omega^*(M,\xi)^{\TT}.\\
\end{split}
\ee 
\end{theorem}
\begin{proof} The first relation can be proved in a verbatim way as the situation without the presence of $\xi$. 
 
 To prove the second relation, we have
 \be 
 \begin{split}
 &\{\nabla^{\xi}-\iota_v+H, \gamma_\U\}\\
 =&\{\nabla^{\xi}-\iota_v+H, \iota_X+\alpha\wedge\}\\
 =&\{\nabla^{\xi}, \iota_X\}+\{\nabla^{\xi}, \alpha\wedge\}-\{\iota_v, \iota_X\}+\{\iota_v, \alpha\wedge\}+\{H, \iota_X\}+\{H,\alpha\wedge\}\\
 =& \{\nabla^{\xi}, \iota_X\}+d\alpha\wedge-\iota_v\alpha\wedge+\iota_XH\wedge\\
 =&L_X^\xi-\mu^\xi_X+(d\alpha-\iota_v\alpha+\iota_X H)\wedge.
 \end{split}
 \ee
 
To prove the third relation, simply notice that equation (\ref{maineqn}) tells us that on $\Omega^*(M,\xi)^{\TT}$, $\{\nabla^{\xi}-\iota_v+H, \nabla^{\xi}-\iota_v+H\}=0$ and apply the second relation. 

To prove the fourth relation, we have
\be
\begin{split}
&[L_X^\xi-\mu^\xi_X+(d\alpha-\iota_v\alpha+\iota_X H)\wedge, \iota_Y+\beta\wedge]\\
=&[L_X^\xi, \iota_Y]+[L_X^\xi, \beta\wedge]-[\mu^\xi_X, \iota_Y]-[\mu^\xi_X, \beta\wedge]+[d\alpha\wedge, \iota_Y]+[d\alpha\wedge, \beta\wedge]\\
&-[\iota_v\alpha\wedge, \iota_Y]-[\iota_v\alpha\wedge, \beta\wedge]+[\iota_XH\wedge, \iota_Y]-[\iota_XH\wedge, \beta\wedge]\\
=&\iota_{[X, Y]}+(L_X\beta)\wedge-0-0-\iota_Y(d\alpha)\wedge+0-0-0-\iota_Y(\iota_XH)-0\\
=&\iota_{[X, Y]}+(L_X\beta-\iota_Y(d\alpha)+\iota_X\iota_YH)\wedge,
\end{split} 
\ee
and this shows that 
$$[\mathcal{L}^\xi_{\U}, \gamma_\V]=\gamma_{\U\circ_H \V}.$$

The last relation can be deduced by combining the second, the third and the fourth relation. 
\end{proof}

Antisymmetrizing the fourth relation, we get that 

\begin{corollary} $\forall \varphi \in \Omega^*(M,\xi)^{\TT}$, the following identity holds, 
\be 
\begin{split}
&\gamma_\U\gamma_\V\cdot ((\nabla^{\xi}-\iota_v+H)\varphi)\\
=&(\nabla^{\xi}-\iota_v+H)(\gamma_\U\gamma_\V\cdot \varphi)+\gamma_\V\cdot ((\nabla^{\xi}-\iota_v+H)(\gamma_\U\cdot \varphi))\\
&-\gamma_\U\cdot((\nabla^{\xi}-\iota_v+H)(\gamma_\V\cdot \varphi))+\gamma_{[[\U, \V]]_H}\cdot\varphi.\\
\end{split}
\ee

\end{corollary}

\section{T-duality and exotic Hori formulae}

\subsection{Review of T-duality} First we review the results
in  \cite{BEM04a, BEM04b}, where the following situation is studied. We give more details here than the brief review in the introduction.

In \cite{BEM04a, BEM04b}, spacetime $Z$ was compactified in one direction.
More precisely, $Z$
is a principal $\bbT$-bundle over $X$

\begin{equation}\label{eqn:MVBx}
\begin{CD}
\bbT @>>> Z \\
&& @V\pi VV \\
&& X \end{CD}
\end{equation}
 classified up to isomorphism by its first Chern class
{ $c_1(Z)\in H^2(X,\ZZ)$}. Assume that spacetime $Z$ is endowed with an $H$-flux which is
a representative in the
degree 3 Deligne cohomology of $Z$, that is
$H\in\Omega^3(Z)$ with integral periods (for simplicity, we drop factors of $\frac{1}{2\pi i}$),
together with
the following data. Consider a local trivialization $U_\alpha \times \TT$ of $Z\to X$, where
$\{U_\alpha\}$ is a good cover of $X$. Let $H_\alpha = H\Big|_{ U_\alpha \times \TT}
= d B_\alpha$, where $B_\alpha \in \Omega^2(U_\alpha \times \TT)$ and finally, $B_\alpha -B_\beta = F_{\alpha\beta}
\in \Omega^1(U_{\alpha\beta} \times \TT)$.
 Then the choice of $H$-flux entails that we are given a local trivialization
 as above and locally defined 2-forms $B_\alpha$ on it, together with closed 2-forms $F_{\alpha\beta}$ defined on double overlaps,  that is, $(H, B_\alpha, F_{\alpha\beta})$. Also the first Chern class
 of $Z\to X$
  is represented in integral cohomology by $(F, A_\alpha)$ where
$\{A_\alpha\}$ is a connection 1-form on $Z\to X$ and $F = dA_\alpha$ is the curvature 2-form of $\{A_\alpha\}$.

The {  {T-dual}}  is another principal
$\bbT$-bundle over $M$, denoted by $\hat Z$,
  {}
\begin{equation}\label{eqn:MVBy}
\begin{CD}
\hat \bbT @>>> \hat Z \\
&& @V\hat \pi VV     \\
&& X \end{CD}
\end{equation}
To define it, we see that $\pi_* (H_\alpha) = d \pi_*(B_\alpha) = d {\hat A}_\alpha$,
 so that $\{{\hat A}_\alpha\}$ is a connection 1-form whose curvature $ d {\hat A}_\alpha = \hat F_\alpha =  \pi_*(H_\alpha)$
 that is, $\hat F = \pi_* H$. So let $\hat Z$ denote the principal
$\bbT$-bundle over $M$ whose first Chern class is  $\,\, c_1(\hat Z) = [\pi_* H, \pi_*(B_\alpha)] \in H^2(X; \ZZ) $.

The Gysin
sequence \cite{BT} for $Z$ enables us to define a T-dual $H$-flux
$[\hat H]\in H^3(\hat Z,\ZZ)$, satisfying
\begin{equation} \label{eqn:MVBc}
c_1(Z) = \hat \pi_* \hat H \,,
\end{equation}
 where $\pi_* $
and similarly $\hat\pi_*$, denote the pushforward maps.
Note that $ \hat H$ is not fixed by this data, since any integer
degree 3 cohomology class on $X$ that is pulled back to $\hat Z$
also satisfies the requirements. However, $ \hat H$ is
determined uniquely (up to cohomology) upon imposing
the condition $[H]=[\hat H]$ on the correspondence space $Z\times_X \hat Z$
as will be explained now.

The {\em correspondence space} (sometimes called the doubled space) is defined as
$$
Z\times_X  \hat Z = \{(x, \hat x) \in Z \times \hat Z: \pi(x)=\hat\pi(\hat x)\}.
$$
Then we have the following commutative diagram,
\begin{equation*} \label{eqn:correspondence}
\xymatrix @=6pc @ur { (Z, [H]) \ar[d]_{\pi} &
(Z\times_X  \hat Z, [H]=[\hat H]) \ar[d]_{\hat p} \ar[l]^{p} \\ X & (\hat Z, [\hat H])\ar[l]^{\hat \pi}}
\end{equation*}
By requiring that
$$
p^*[H]={\hat p}^*[\hat H] \in H^3(Z\times_X  \hat Z, \ZZ),
$$
determines $[\hat H] \in H^3(  \hat Z, \ZZ)$  uniquely, via an application of the Gysin sequence.
An alternate way to see this is is explained below.

Let $(H, B_\alpha, F_{\alpha\beta}, L_{\alpha\beta})$ denote a gerbe with connection on $Z$, cf. \cite{Brylinski}, where 
$(H, B_\alpha, F_{\alpha\beta})$ denotes the Deligne class of the closed integral 3-form $H$ and $ L_{\alpha\beta}$
denotes the line bundles on double overlaps that determines the gerbe.
We also choose a connection 1-form $A$ on $Z$.
Let $v$ denote the vector field generating the $\TT$-action on $Z$.
Then define $\widehat A_\alpha = -\imath_v B_\alpha$ on the chart $U_\alpha$ and
the connection 1-form $\widehat A= \widehat A_\alpha +d\widehat\theta_\alpha$
on the chart $U_\alpha\times  \widehat \TT$. In this way we get a T-dual circle bundle
$\widehat Z \to X$ with connection 1-form $\widehat A$.

Without loss of generality, we can assume that $H$ is $\TT$-invariant. Consider
$$
\Omega = H - A\wedge F_{\widehat A}
$$
where  $F_{\widehat A} = d {\widehat A}$ and $F_{A} = d {A}$ are the curvatures of $A$
and $\widehat A$ respectively. One checks that the contraction $i_v(\Omega)=0$ and
the Lie derivative $L_v(\Omega)=0$ so that $\Omega$ is a basic 3-form on $Z$, that is
$\Omega$ comes from the base $X$.

Setting
$$
\widehat H = F_A\wedge {\widehat A} + \Omega
$$
this defines the T-dual flux 3-form. One verifies that $\widehat H$ is a closed 3-form on $\widehat Z$.
It follows that on the correspondence space, one has as desired,
\begin{equation}
\widehat H = H + d (A\wedge \widehat A ).
\end{equation}

Our next goal is to determine the T-dual curving or B-field.
The Buscher rules imply that on the open sets $U_\alpha \times \TT\times \widehat \TT$ of the
correspondence space $Z\times_X \hat Z$, one has
\begin{equation}
\widehat B_\alpha = B_\alpha + A\wedge \widehat A - d\theta_\alpha \wedge d\widehat \theta _\alpha\,,
\end{equation}
Note that
\begin{equation}
\imath_v \widehat B_\alpha = \imath_v
\left( B_\alpha + A\wedge \widehat A - d\theta_\alpha \wedge d\widehat \theta _\alpha\right) =
-\widehat A_\alpha + \widehat A - d\widehat \theta_\alpha = 0
\end{equation}
so that $\widehat B_\alpha$ is indeed a 2-form on $\widehat Z$ and not just on the correspondence
space. Obviously, $d \widehat B_\alpha = \widehat H$. Following the descent equations one arrives at the complete
T-dual gerbe with connection, $(\widehat H, \widehat B_\alpha, \widehat F_{\alpha\beta}, \widehat L_{\alpha\beta})$.
cf. \cite{BMPR}.

Define the Riemannian metrics on $Z$ and $\hat Z$ respectively by
$$
g=\pi^*g_X+R^2\, A\odot A,\qquad \hat g=\hat\pi^*g_X+1/{R^2} \,\hat A\odot\hat A.
$$
where $g_X$ is a Riemannian metric on $X$.
 Then $g$ is $\TT$-invariant and the length of each circle fibre is $R$; $\hat g$
 is $\hat\TT$-invariant and the length of each circle fibre is $1/R$.

The rules
for transforming the Ramond-Ramond (RR) fields can be encoded in the {\cite{BEM04a, BEM04b}} generalization of
{\em Hori's formula}
  {}
\begin{equation} \label{eqn:Hori}
T_*G =  \int^{\bbT} e^{ -A \wedge \hat A }\ G \,,
\end{equation}
 where $G \in \Omega^\bullet(Z)^\bbT$ is the total RR field-strength,
\begin{center}
$G\in\Omega^{even}(Z)^\bbT \quad$ for {   { Type IIA}};\\
$G\in\Omega^{odd}(Z)^\bbT \quad$ for {   { Type IIB}},\\
\end{center}
and where the right hand side of equation \eqref{eqn:Hori} is an invariant differential form on $Z\times_X\hat Z$, and
the integration is along the $\bbT$-fiber of $Z$.

Recall that the twisted cohomology
is defined as the cohomology of the complex
\be H^\bullet(Z, H) = H^\bullet(\Omega^\bullet(Z), d_H=d+ H\wedge).\ee
By the identity \eqref{eqn:Hori}, $T_*$ maps $d_H$-closed forms $G$ to $d_{\hat
H}$-closed forms $T_*G$.
 So T-duality $T_*$  induces a map on twisted cohomologies,
$$
T : H^\bullet(Z, H) \to H^{\bullet +1}(\hat Z, \hat H).
$$


\subsection{Exotic theories}
Let $\xi, \hat{\xi}$ be the complex line bundle determined by the circle bundles $Z, \hat{Z}$ and the standard representation of the circle on the complex plane. Let $\nabla^\xi$ and $\nabla^{\hat{\xi}}$ be the connections on $\xi, \hat{\xi}$ induced from the connections on $Z, \hat{Z}$ respectively. Let $v, \hat{v}$ be the vertical tangent vector fields on $Z, \hat{Z}$ respectively as in the above section.
\begin{theorem} \label{flat} $\forall n\in \mathbb{Z}$, we have: \newline
on $\Omega^*(Z, \pi^*(\hat{\xi}^{\otimes n}))$, the following identity holds,
\be (\pi^*\nabla^{\hat{\xi}^{\otimes n}}-\iota_{nv}+H)^2+nL_v^{\hat{\xi}^{\otimes n}}=0; \ee
on $\Omega^*(\hat{Z}, \pi^*(\xi^{\otimes n}))$, the following identity holds,
\be (\hat{\pi}^*\nabla^{\xi^{\otimes n}}-\iota_{n\hat{v}}+\hat{H})^2+nL_{\hat{v}}^{\xi^{\otimes n}}=0. \ee
\end{theorem}
\begin{proof} Let $\{\hat{s}_\alpha\}$ be local sections of of the line bundle $\hat{\xi}$ such that the connection 1-form corresponding to whom is $\{\hat{A}_\alpha\}.$ It is obvious that $\pi^*\hat{s}_\alpha$'s are invariant about the $\TT$-action on $Z$. Under $\{\pi^*\hat{s}_\alpha\}$, to prove the first equality, we only need to prove that for forms on $U_\alpha\times \TT$, 
\be (d+\pi^*(n\hat{A}_\alpha)-\iota_{nv}+H)^2+ nL_v=0.    \ee
But this is evident, as we have
$$(d-\iota_{nv})^2+nL_v=0, $$
$$dH=0, \ \ \iota_v\pi^*(\hat{A}_\alpha)=0, $$
$$ d\pi^*(n\hat{A}_\alpha)-n\iota_vH=n\pi^*(d\hat{A}_\alpha)-n\iota_vH=0.$$

One can similarly prove the second equality. 
\end{proof}

Denote
$$\mathcal{A}^*(Z)=\bigoplus_{n\in \ZZ} \mathcal{A}_n^*(Z):=\bigoplus_{n\in \ZZ}\Omega^*(Z, \pi^*(\hat{\xi}^{\otimes n}))^{\TT}. $$ 
$$\mathcal{A}^*(\hat Z)=\bigoplus_{n\in \ZZ} \mathcal{A}_n^*(\hat Z):=\bigoplus_{n\in \ZZ}\Omega^*(\hat Z, \hat \pi^*(\xi)^{\otimes n}))^{\hat \TT}. $$
For each $n\in \ZZ$, consider the complexes
$$(\mathcal{A}_n^*(Z), \pi^*\nabla^{\hat{\xi}^{\otimes n}}-\iota_{nv}+H),$$
$$ (\mathcal{A}_n^*(\hat Z), \hat{\pi}^*\nabla^{\xi^{\otimes n}}-\iota_{n\hat{v}}+\hat{H}).$$
We have
\begin{theorem} \label{concentration} If $n\neq 0$, then  
\be  H(\mathcal{A}_n^*(Z), \pi^*\nabla^{\hat{\xi}^{\otimes n}}-\iota_{nv}+H)\cong \{0\},   \ee
\be  H(\mathcal{A}_n^*(\hat Z), \hat{\pi}^*\nabla^{\xi^{\otimes n}}-\iota_{n\hat{v}}+\hat{H})\cong \{0\}.   \ee
\end{theorem}
\begin{proof}  For $n\neq 0$, set $\eta_n=\frac{A}{F_A-n}$, where $A$ is the connection 1-form on $Z$ and $F_A$ is its curvature 2-form. 
Then we have
\be 
\begin{split}
&(d-\iota_{nv})\eta_n\\
=&\frac{[(d-\iota_{nv})A](F_A-n)-A[(d-\iota_{nv})(F_A-n)]}{(F_A-n)^2}\\
=&\frac{(F_A-n)^2}{(F_A-n)^2}\\
=&1.
\end{split}
\ee

We therefore obtain a homotopy for $n\neq 0$:  $\forall x\in \Omega^*(Z, \pi^*(\hat{\xi}^{\otimes n})).$
\be 
\begin{split}
&(\pi^*\nabla^{\hat{\xi}^{\otimes n}}-\iota_{nv}+H)(\eta_nx)+\eta_n[(\pi^*\nabla^{\hat{\xi}^{\otimes n}}-\iota_{nv}+H)x]\\
=&[(d-\iota_{nv})\eta_n]x\\
=&x.
\end{split}
\ee
We can prove the second isomorphism in a verbatim way by using $\hat{A}, F_{\hat{A}}$ to to construct the homotopy.

\end{proof}

\subsection{Twisted integration along the fiber}

In order to define the exotic Hori formula to be introduced in the next subsection, we first introduce a twisted version of integration along the fiber.

Let $\pi: P\to M$ be a principal circle bundle over $M$, and $\Theta$ a connection one form on $P$. Let $L$ be a Hermitian line bundle over $M$ such that $Z$ is the circle bundle of $L$. Let $\nabla^L$ be the connection on $L$ corresponding to $\Theta$. Choose a good cover $\{U_\alpha\}$ on $M$ such that $\pi^{-1}(U_\alpha)\cong U_\alpha\times S^1.$ Let $\{f_\alpha\}$ be a local basis of $L$ corresponding to the constant map $U_\alpha\to \{1\}\subset S^1.$ 

$\forall n\in \ZZ$, define the {\bf twisted integration along the fiber} as follows: for $\omega\in \Omega^*(P)$, 
\be  \int^{P/M, n}\!\!\!\!\!\omega\in \Omega^*(M, L^{\otimes n}), \ \mathrm{such\ that}\   \left.\left(\int^{P/M, n}\!\!\!\!\! \omega\right)\right|_{U_\alpha}=\left(\int^{\pi^{-1}(U_\alpha)/U_\alpha}\omega_\alpha e^{2\pi in\theta_\alpha}\right)\otimes f_\alpha^{\otimes n},   \ee
where $\omega_\alpha=\omega|_{\pi^{-1}(U_\alpha)}$, $\theta_\alpha$ is the vertical coordinates of $\pi^{-1}(U_\alpha).$ Note that as 
on $U_{\alpha\beta}=U_\alpha\cap U_\beta$,  $f_\alpha/f_\beta=e^{2\pi i(\theta_\beta-\theta_\alpha)}$ (a function on $U_{\alpha\beta}$), the above construction patches to be a global section of the bundle $\wedge^*(T^*M)\otimes L^{\otimes n}$. Moreover, it is not hard to see that this definition is independent of choice of the good cover $\{U_\alpha\}$ and local trivializations.

\begin{theorem} \label{interchange} Let $Y$ be a vector field on $M$ and \, $\widetilde{Y}$ a lift of $Y$ on $P$. Let $H$ be a differential form on $M$. Then $\forall n\in \ZZ$
\be(\nabla^{L\otimes n}-\iota_Y+H)\int^{P/M, n}\!\!\!\!\!\omega=-\int^{P/M, n}\!\!\!\!\! (d+n\Theta-\iota_{\widetilde{Y}}+H)\omega.  \ee

\end{theorem}
\begin{proof} For the definition of the usual integration along the fiber, we refer to \cite{BGV}. It is not hard to see that 
\be H\cdot  \intomega=-\int^{P/M, n}\!\!\!H\cdot \omega,\ee

\be \iota_Y\intomega=-\int^{P/M, n}\!\!\!\iota_{\widetilde{Y}}\omega.    \ee

One only needs to prove 
\be\nabla^{L\otimes n}\left(\intomega\right)=-\int^{P/M, n}\!\!\!(d+n\Theta)\omega.\ee
In fact, suppose $\Theta_\alpha$ is the connection 1-form under the local basis $f_\alpha$, we have
\be
\begin{split}
&\left.\left(\nabla^{L\otimes n}\left(\intomega\right)\right)\right|_{U\alpha}\\
=&\left[d\int^{\pi^{-1}(U_\alpha)/U_\alpha}\omega_\alpha e^{2\pi in\theta_\alpha}+(-1)^{\mathrm{deg}{\omega}-1}\left(\int^{\pi^{-1}(U_\alpha)/U_\alpha}\omega_\alpha e^{2\pi in\theta_\alpha}\right)n\Theta_\alpha\right]\otimes f_\alpha^{\otimes n}\\
=&\left[d\int^{\pi^{-1}(U_\alpha)/U_\alpha}\omega_\alpha e^{2\pi in\theta_\alpha}-\int^{\pi^{-1}(U_\alpha)/U_\alpha}n\Theta_\alpha \omega_\alpha e^{2\pi in\theta_\alpha}\right]\otimes f_\alpha^{\otimes n}\\
=&\left[-\int^{\pi^{-1}(U_\alpha)/U_\alpha}d(\omega_\alpha e^{2\pi in\theta_\alpha})-\int^{\pi^{-1}(U_\alpha)/U_\alpha}n\Theta_\alpha \omega_\alpha e^{2\pi in\theta_\alpha}\right]\otimes f_\alpha^{\otimes n}\\
=&\left[-\int^{\pi^{-1}(U_\alpha)/U_\alpha}\left((d\omega_\alpha) e^{2\pi in\theta_\alpha}+(-1)^{\mathrm{deg}\omega}\omega_\alpha e^{2\pi in\theta_\alpha}(2\pi i nd\theta_\alpha)\right)\right.\\
&\left.\ \ -\int^{\pi^{-1}(U_\alpha)/U_\alpha}n\Theta_\alpha \omega_\alpha e^{2\pi in\theta_\alpha}\right]\otimes f_\alpha^{\otimes n}\\
=&-\left[\int^{\pi^{-1}(U_\alpha)/U_\alpha}[(d\omega_\alpha+n\Theta_\alpha+n2\pi id\theta_\alpha)\omega_\alpha] e^{2\pi in\theta_\alpha}\right]\otimes f_\alpha^{\otimes n}\\
=&-\left.\left(\int^{P/M, n} \!\!\!\!(d+n\Theta)\omega\right)\right|_{U\alpha}.\\
\end{split}
\ee
The desired equality follows.

\end{proof}

\subsection{The exotic Hori formulae}

Let us go back to the T-duality with same notions as in Section 2.1, 2.2.

Let $X+\alpha\in \Gamma(TZ\oplus T^*Z)^{\TT}$. Then one can write
\be X=x+fv, \ \alpha=\theta+gA,\ee
where $x\in \Gamma(TM), \theta\in \Omega^1(M), f, g\in C^\infty(M)$. Define 
$$\phi(X, \alpha)=(x+gv)+(\theta+fA). $$

Recall that
\be \Omega^*_{-n}(Z)=\{\omega \in \Omega^*(Z)| L_{v}\omega=-n\omega\}.\ee
Note that $\Omega^*_{-n}(Z)=\Omega^*(Z)^\TT$, i.e. the $\TT$-invariant forms on $Z$. 

Let $\omega_{-n}\in \Omega^*_{-n}(Z).$ Define the {\bf exotic Hori formula} by
\be\label{exotic-Hori} \tau_n(\omega_{-n})= \int^{\TT, n}\omega_{-n}\, e^{-A\wedge\hat{A}}\in \Omega^*(\hat Z, \hat\pi^*(\xi)^{\otimes n})),\ee
where $\int^{\TT, n}$ stands for $\int^{(Z\times_X \hat{Z})/\hat{Z}, n}$ for simplicity. 
\begin{remark}(i) Let $\{s_\alpha\}$ be local sections of of the line bundle $\xi$ and $\theta_\alpha$ be the vertical coordinate function on $\pi^{-1}(U_\alpha)$. Then, locally, $\omega_{-n}$ must be of the form
$$(\omega_{-n, \alpha, 0}+ \omega_{-n, \alpha, 1}(d\theta_\alpha+A_\alpha))e^{-2\pi in\theta_\alpha},$$ where $\omega_{-n, \alpha, 0}$ and  $\omega_{-n, \alpha, 1}$ are both forms on $U_\alpha$. So if $m\neq n$, 
\be 
\begin{split}
&\left.\int^{\TT, m}\omega_{-n}\, e^{-A\wedge\hat{A}}\ \right|_{U_\alpha}\\
=&\left( \int^{\TT}(\omega_{-n, \alpha, 0}+ \omega_{-n, \alpha, 1}(d\theta_\alpha+A_\alpha))(1-(d\theta_\alpha+A_\alpha)\wedge\hat{A})e^{-2\pi in\theta_\alpha}\cdot e^{2\pi i m\theta_\alpha}\right)\otimes\hat\pi^*(s_\alpha)^{\otimes n}\\
=&0.
\end{split}
\ee
This explains why we only define $\tau_m(\omega_{-n})$ for $m=n$.

(ii) When $n=0$, $\tau_0$ is just the Hori formula (\ref{eqAh}) in {\cite{BEM04a, BEM04b}}.
\end{remark}

Denote by $\rho$ the tautological global section of the line bundle $(\hat\pi\circ \hat p)^*\xi$ on $Z\times_X \hat{Z}$. Let $\hat \theta_n \in \Omega^*(\hat Z, \hat\pi^*(\xi)^{\otimes n})^{\hat \TT}$. Define the {\bf inverse exotic Hori formula} by
\be \hat{\sigma}_n(\hat \theta_{n})= \int^{\hat\TT}\hat p^*(\hat \theta_n) \cdot (\rho^{-1})^{\otimes n} \cdot e^{A\wedge\hat{A}}\in \Omega^*(Z). \ee

One can dually define the exotic Hori formula $\hat\tau_n$ from $\hat Z$ to $Z$ and the inverse exotic Hori formula $\sigma_n$ from $Z$ to $\hat Z$. Let $\hat \omega_{-n}\in \Omega^*_{-n}(\hat Z).$ Define 
\be \hat \tau_n(\hat \omega_{-n})= \int^{\hat \TT, n}\hat\omega_{-n}\, e^{A\wedge\hat{A}}\in \Omega^*(Z, \pi^*(\hat\xi)^{\otimes n})).\ee Denote by $\hat\rho$ the tautological global section of the line bundle $(\pi\circ p)^*\hat\xi$ on $Z\times_X \hat{Z}$. Let $\theta_n \in \Omega^*(Z, \pi^*(\hat\xi)^{\otimes n})^{\TT}$. Define 
\be \label{exotic-Hori2}\sigma_n(\theta_{n})= \int^{\TT}p^*(\theta_n) \cdot (\hat\rho^{-1})^{\otimes n} \cdot e^{-A\wedge\hat{A}}\in \Omega^*(\hat Z). \ee

We have the following results:

\begin{theorem}\label{main}
(1) $\phi$ is orthogonal with respect to the pairing on $TZ\oplus T^*Z$, hence induces an isomorphism on Clifford algebras.\newline
(2) $\phi$ preserves the twisted Courant bracket.\newline
(3) $\tau_n(\omega_{-n})\in  \Omega^*(\hat{Z}, \hat{\pi}^*(\xi^{\otimes n}))^{\hat \TT}$ and $\hat{\sigma}_n(\theta_{n})\in \Omega^*_{-n}(Z).$ \, For $\mathcal{U}\in\Gamma(TZ\oplus T^*Z)^{\TT}$, we have $\tau_n(\gamma_{\mathcal{U}}\cdot \{\omega_{-n}\})=\gamma_{\phi(\mathcal{U})}\cdot \tau_n(\{\omega_{-n}\})$, for all $\{\omega_{-n}\}\in \Omega^*_{-n}(Z),$ hence induces an isomorphism of Clifford modules 
$$\tau_n: \Omega^*_{-n}(Z)\to \Omega^*(\hat{Z}, \hat{\pi}^*(\xi^{\otimes n}))^{\hat \TT}.$$
$\hat \sigma_n=-\tau_n^{-1}$ and is an isomorphism of Clifford modules $$\hat \sigma_n: \Omega^*(\hat{Z}, \hat{\pi}^*(\xi^{\otimes n}))^{\hat \TT}\to \Omega^*_{-n}(Z).$$ The dual results for $\hat\tau_n$ and $\sigma_n$ are also true. \newline
(4) The map $\tau_n$ induces a chain map on the complexes 
$$(\Omega^*_{-n}(Z), d+H)\to (\Omega^*(\hat{Z}, \hat{\pi}^*(\xi^{\otimes n}))^{\hat \TT},-(\hat{\pi}^*\nabla^{\xi^{\otimes n}}-\iota_{n\hat{v}}+\hat{H}))$$ and the map $\hat \sigma_n$ induces induces a chain map on the complexes 
$$(\Omega^*(\hat{Z}, \hat{\pi}^*(\xi^{\otimes n})),\hat{\pi}^*\nabla^{\xi^{\otimes n}}-\iota_{n\hat{v}}+\hat{H})^{\hat \TT}\to(\Omega^*_{-n}(Z), -(d+H)).$$ The dual results for $\hat\tau_n$ and $\sigma_n$ are also true. 
\end{theorem}

(1), (2) in the above theorem as well as (3) and (4) without the presence of the line bundles $\xi, \hat{\xi}$ are existing results (\cite{Gua03, Gua07}, c.f. \cite{B10}).

\begin{proof} We will prove (3) and (4). 


Let $\{s_\alpha\}$ be local sections of of the line bundle $\xi$ and $\{\hat{s}_\alpha\}$ local sections of of the line bundle $\hat{\xi}$. Let $\theta_\alpha$ be the vertical coordinate function of $\pi^{-1}(U_\alpha)$ and $\hat{\theta}_\alpha$ the vertical coordinate function of $\hat{\pi}^{-1}(U_\alpha)$. 

Take $\{\omega_n\}\in \Omega^*_{-n}(Z)$. Locally, $\omega_{-n}$ is of the form
$$(\omega_{-n, \alpha, 0}+ \omega_{-n, \alpha, 1}A)e^{-2\pi in\theta_\alpha},$$ where $\omega_{-n, \alpha, 0}$ and  $\omega_{-n, \alpha, 1}$ are both forms on $U_\alpha$. Then 
\be \label{tau}
\begin{split}
&\left.\tau_n(\omega_{-n})\,\right|_{U_\alpha}\\
=&\left.\int^{\TT, n}\omega_n\, e^{-A\wedge\hat{A}}\, \right|_{U_\alpha}\\
=&\left( \int^{\TT}(\omega_{-n, \alpha, 0}+ \omega_{-n, \alpha, 1}A)e^{-A\wedge\hat{A}}e^{-2\pi in\theta_\alpha}\cdot e^{2\pi i n\theta_\alpha}\right)\otimes\hat\pi^*(s_\alpha)^{\otimes n}\\
=&\left( \int^{\TT}(\omega_{-n, \alpha, 0}+ \omega_{-n, \alpha, 1}A)e^{-A\wedge\hat{A}}\right)\otimes\hat\pi^*(s_\alpha)^{\otimes n},
\end{split}
\ee
which is equal to
$$(-\omega_{-n, \alpha, 0}\hat A-\omega_{-n, \alpha, 1})\otimes\hat\pi^*(s_\alpha)^{\otimes n}$$
if $\omega_{-n}$ is of even degree; or
$$(\omega_{-n, \alpha, 0}\hat A+\omega_{-n, \alpha, 1})\otimes\hat\pi^*(s_\alpha)^{\otimes n}$$
if $\omega_{-n}$ is of odd degree. So we see that
$$ \tau_n(\omega_{-n}) \in  \Omega^*(\hat{Z}, \hat{\pi}^*(\xi^{\otimes n}))^{\hat \TT}. $$

On the other hand, if $\hat \theta_n\in \Omega^*(\hat Z, \hat\pi^*(\xi)^{\otimes n}))^{\hat \TT}$, suppose $\hat\theta_n$ is locally equal to
$$(\hat\eta_{n, \alpha, 0}\hat A+\hat\eta_{n, \alpha, 1})\otimes\hat\pi^*(s_\alpha)^{\otimes n}, $$
where $\hat\eta_{n, \alpha, 0}, \hat\eta_{n, \alpha, 1}$ are forms on $U_\alpha$.
then $\hat{\sigma}_n(\hat\theta_{n})$ is locally equal to 
\be \label{T hat}
\left.\int^{\hat\TT}\hat p^*(\hat\theta_n) \cdot (\rho^{-1})^{\otimes n} \cdot e^{A\wedge\hat{A}}\, \right|_{U_\alpha}=\left(\int^{\hat\TT}(\hat\eta_{n, \alpha, 0}\hat A+\hat\eta_{n, \alpha, 1}) \cdot e^{A\wedge\hat{A}}\right)\cdot e^{-2\pi in \theta_\alpha}, 
\ee which is equal to
$$(\hat\eta_{n, \alpha, 0}+\hat\eta_{n, \alpha, 1}A)\cdot e^{-2\pi in \theta_\alpha} $$if $\hat\theta_n$ is of even degree; or
$$ -(\hat\eta_{n, \alpha, 0}+\hat\eta_{n, \alpha, 1}A)\cdot e^{-2\pi in \theta_\alpha} $$if $\hat\theta_n$ is of odd degree.
And so evidently $L_v \hat{\sigma}_n(\hat\theta_{n})=-n\hat{\sigma}_n(\hat\theta_{n})$. We therefore have
$$ \hat{\sigma}_n(\hat\theta_{n})\in \Omega^*_{-n}(Z). $$

From the above expressions (\ref{tau}), (\ref{T hat}) and local nature of (3), we see from the original Hori formula that $\tau_n, \hat \sigma_n$ both respect the Clifford actions and 
\be \hat\sigma_n=-\tau_n^{-1}. \ee

We next prove (4). We have 
\be [(d+H)\omega_{-n}]e^{-A\wedge\hat{A}}=[d+H-(H-\hat H)](\omega_{-n}e^{-A\wedge\hat{A}})=(d+\hat H)(\omega_{-n}e^{-A\wedge\hat{A}}). \ee
Also one has
\be (nA-\iota_{n\hat{v}})(\omega_{-n}e^{-A\wedge\hat{A}})=nA\omega_{-n}e^{-A\wedge\hat{A}}-\iota_{n\hat{v}}(\omega_{-n}e^{-A\wedge\hat{A}})=(-1)^{|\omega_{-n}|}\omega_{-n}(nA-nA)=0.  \ee

By Theorem \ref{interchange}, we have
\be
\begin{split}
&\tau_n((d+H)\omega_{-n})\\
=&\int^{\TT, n}[(d+H)\omega_{-n}]\cdot e^{-A\wedge\hat{A}}\\
=&\int^{\TT, n}(d+\hat H)(\omega_{-n}\cdot e^{-A\wedge\hat{A}})\\
=&\int^{\TT, n}(d+nA-\iota_{n\hat{v}}+\hat H)(\omega_{-n}\cdot e^{-A\wedge\hat{A}})\\
=&-(\hat{\pi}^*\nabla^{\xi^{\otimes n}}-\iota_{n\hat{v}}+\hat{H})\int^{\TT, n} \omega_{-n}\cdot e^{-A\wedge\hat{A}} \\
=&-(\hat{\pi}^*\nabla^{\xi^{\otimes n}}-\iota_{n\hat{v}}+\hat{H})\tau_n(\omega_{-n}).
\end{split}
\ee
As $-\hat \sigma_n$ is the inverse of $\tau_n$ , one deduces that $\hat \sigma_n$ is also a chain map. 

\end{proof}

Combining Theorem \ref{concentration}, we have 
\begin{corollary}\label{nullcoho} \label{nullcoho} If $n\neq 0$, then
\be H(\Omega_{-n}^*(Z), d+H)=0,\ee
\be H(\Omega_{-n}^*(\hat Z), d+\hat H)=0.\ee
\end{corollary}
Actually, from the proof of Theorem \ref{concentration}, we have the following homotopy 
\be \label{homotopy} (d+H)\hat\sigma_n[\hat\eta_n\cdot \tau_n(\omega_{-n})]+\hat \sigma_n(\hat\eta_n\cdot \tau_n[(d+H)\omega_{-n}])=\omega_{-n},\ee
where $\omega_{-n}\in \Omega_{-n}^*(Z)$ and $\hat\eta_n=\frac{\hat A}{F_{\hat A}-n}$. It is interesting to see that to construct this homotopy on $(\Omega_{-n}^*(Z), d+H)$, one uses the data on the dual side $\hat{Z}$. The interested readers may write a homotopy for the dual case.

For a general form $\omega\in \Omega^*(Z)$, one can perform the {\bf family Fourier expansion} as follows. 
Suppose locally on $U_\alpha$, 
\be \omega=\sum_I f_I(x, \theta_\alpha)dx_I+\sum_Jg_J(x, \theta_\alpha)dx_J\wedge A.\ee
Consider the local form
$$\sum_I \left(\int^{\TT}f_I(x, \theta_\alpha )e^{2\pi in\theta_\alpha}d\theta_\alpha\right)dx_I+\sum_J\left(\int^{\TT}g_J(x, \theta_\alpha )e^{2\pi in\theta_\alpha}d\theta_\alpha\right)dx_J\wedge A.$$
As $e^{2\pi in(\theta_\alpha-\theta_\beta)}$ is a function on the base, the form
$$\left[\sum_I \left(\int^{\TT}f_I(x, \theta_\alpha )e^{2\pi in\theta_\alpha}d\theta_\alpha\right)dx_I+\sum_J\left(\int^{\TT}g_J(x, \theta_\alpha )e^{2\pi in\theta_\alpha}d\theta_\alpha\right)dx_J\wedge A\right] e^{-2\pi in\theta_\alpha}$$ for each $\alpha$ glue together to be a global form on $Z$. Denote this form by $\omega_{-n}$. Since $\left(\int^{\TT}f_I(x, \theta_\alpha )e^{2\pi in\theta_\alpha}d\theta_\alpha\right)dx_I$, $\left(\int^{\TT}g_J(x, \theta_\alpha )e^{2\pi in\theta_\alpha}d\theta_\alpha\right)dx_J$ as well as $A$ are all $\TT$-invariant, we see that $L_{v}\omega_{-n}=-n\omega_{-n}$. By the Fourier expansion, $\omega=\sum_{n=-\infty}^\infty \omega_{-n}.$ 

{ Let ${s_\alpha}$ be the local basis of the bundle $\xi$. Then 
\be \left[\sum_I \left(\int^{\TT}f_I(x, \theta_\alpha )e^{2\pi in\theta_\alpha}d\theta_\alpha\right)dx_I+\sum_J\left(\int^{\TT}g_J(x, \theta_\alpha )e^{2\pi in\theta_\alpha}d\theta_\alpha\right)dx_J\wedge A\right]\otimes {\pi^*s_\alpha}^{\otimes n} \ee
patch together to be an element $\Omega_{-n}\in \Omega^*(Z, \pi^*\xi^{\otimes n})$. Let $\gamma$ be the tautological global section of the bundle $\pi^*\xi$ over $Z$. Then 
\be \omega=\sum_{n=-\infty}^\infty \omega_{-n}=\sum_{n=-\infty}^\infty \Omega_{-n}\otimes (\gamma^{-1})^{\otimes n}.\ee
We call $\Omega_{-n}$'s the {\bf family Fourier coefficients} of $\omega$. 

The above theorem shows that 
$$\tau_n(\omega_{-n})=\tau_n(\Omega_{-n}\otimes (\gamma^{-1})^{\otimes n})\in \Omega^*(\hat Z, \hat\pi^*(\xi)^{\otimes n}))^{\hat \TT}.$$
The  $n$ on the left hand side of the above equality should be the momentum of $Z$, as it is $n$-th power of $\xi$. Let 
$$\omega_{-n}-(\iota_v \omega_{-n})A\neq 0,$$ i.e. $\omega_{-n}$ is not some kind of product. 
Suppose $$(d+H)(\omega_{-n})=0.$$
Then from the proof of Theorem \ref{main}, we see that
\be (\hat{\pi}^*\nabla^{\xi^{\otimes n}}-\iota_{m\hat{v}}+\hat{H})\tau_n(\omega_{-n})=0 \iff m=n, \ee
where $m$ is the winding of $\hat Z$ as it the multiple of $\hat v$. This clearly also applies for the trivial bundle case.
}

\subsection{The trivial bundles case}
Consider the trivial bundles case. Now 
$$Z=M\times \TT, \ \ \hat Z=M\times \hat\TT$$ and $H, \hat H$ and the connections are all 0. 

Pick $\omega_{-n}\in \Omega^{even}(Z)_{-n}$. It is of the form $(\lambda_0+\lambda_1d\theta)e^{-2\pi i n\theta}$, where $\lambda_0, \lambda_1$ are forms on $M$. Then by definition, 
$$\tau_n (\omega_{-n})=-\lambda_0d\hat\theta-\lambda_1, \ \ \hat \sigma_n(-\lambda_0d\hat\theta-\lambda_1)=-\lambda_0-\lambda_1d\theta.$$

Suppose $d\omega_{-n}=0$. 

We have 
$$d\lambda_0=0, \ \ d\lambda_1-n\lambda_0=0.$$
Then
$$(d-\iota_{n\hat v} )\tau_n (\omega_{-n})=-(d-\iota_{n\hat v} )(\lambda_0d\hat\theta+\lambda_1)=-(d\lambda_1-n\lambda_0)=0,$$
i.e. $\tau_n (\omega_{-n})$ is exotic equivariant closed (in this case equivariant closed).

If $n\neq 0$, the homotopy (\ref{homotopy}) shows that 
$$d\left(\frac{1}{n}\lambda_1e^{-2\pi i n\theta} \right)=(\lambda_0+\lambda_1d\theta)e^{-2\pi i n\theta}=\omega_{-n},$$
i.e. $\omega_{-n}$ is $(d+H)$-exact (in this case $d$-exact).

One can similarly do the odd degree case.

\end{document}